\documentclass[10pt]{article}
\usepackage[utf8]{inputenc}
\usepackage[T1]{fontenc}

\usepackage{amsmath}
\usepackage{amssymb}
\usepackage{amsthm}
\usepackage[a4paper,left=2cm,top=2cm,right=2cm,bottom=2cm]{geometry}
\usepackage{mathtools}
\usepackage{braket}
\usepackage{tikz}
\usepackage[frozencache]{minted}

\usetikzlibrary{calc,positioning,shapes.geometric,shapes.misc}

\newtheorem{corollary}{Corollary}
\newtheorem{definition}{Definition}
\newtheorem{example}{Example}
\newtheorem{lemma}{Lemma}
\newtheorem{theorem}{Theorem}
\DeclareMathOperator{\ite}{ite}

\title{The multiplicative complexity of interval checking}
\author{Thomas Häner and Mathias Soeken \\ \normalsize Microsoft Quantum, Switzerland}
\date{}

\begin{document}

\maketitle

\begin{abstract}
We determine the exact AND-gate cost of checking if $a\leq x < b$, where $a$ and $b$ are constant integers. Perhaps surprisingly, we find that the cost of interval checking never exceeds that of a single comparison and, in some cases, it is even lower.
\end{abstract}

\section{Introduction}
The multiplicative complexity of a Boolean function $f$ is the smallest number of AND gates in any logic network over the gate set $\{\text{AND, XOR, NOT}\}$ that implements $f$. Multiplicative complexity is used as an important characteristic metric to measure the cost of cryptographic implementations in secure computation protocols~\cite{ARS+15,CDG+17,GMO16} or the cost of fault-tolerant implementations of quantum operations~\cite{MSC+19}.  Unfortunately, computing the multiplicative complexity is intractable~\cite{Find14} and for a random $n$-variable Boolean function $f$ it is at least~$2^{n/2} - O(n)$ with high probability~\cite{BPP00}.  However, for several families of Boolean functions the multiplicative complexity has been analyzed, including quadratic functions~\cite{MS87}, all functions up to 6 variables~\cite{CTP19}, all functions with a multiplicative complexity of at most~4~\cite{CTP20}, all symmetric functions~\cite{BCTP19}, and the Hamming weight function~\cite{BP05}.

In this paper, we determine the multiplicative complexity of the interval check $[a \le x < b]$, where $a$ and $b$ are two nonnegative constant integers and $x$ is an $n$-bit nonnegative integer.  We derive an upper bound on the multiplicative complexity by proposing a construction to implement the interval check, and we derive a matching lower bound based on the algebraic degree of the function (i.e., the largest monomial in its algebraic normal form), thus proving that our construction is optimal with respect to the number of AND gates. We state our main result in Theorem~\ref{thm:main}.



\begin{theorem}[Main result]
\label{thm:main}
Let $n>0$ be the number of bits, and let $a,b$ be two constant integers $a < b < 2^n$. Let $j_a$ and $j_{b}$ denote the number of trailing zeros\footnote{The number of trailing zeros in the binary representation of $a\geq 0$ is the largest integer $j \leq n$ such that $\frac a{2^j}$ is an integer.} in the binary representation of $a$ and $b$, respectively.  Then, the interval check $a \le x < b$ for some arbitrary $n$-bit nonnegative integer $x$ has a multiplicative complexity of
\begin{equation}
    \begin{cases}
        n - \min\{j_a, j_{b}\} - 1 & \text{if $j_a \neq j_{b}$,} \\
        n - j_a - 2 & \text{otherwise}.
    \end{cases}
\end{equation}
\end{theorem}

We present our method for interval checking using an optimal number of AND gates in Section~\ref{sec:construction}. We then analyze our construction and give a proof of Theorem~\ref{thm:main} in Section~\ref{sec:analysis}.

\section{Background}\label{sec:background}
In this section, we introduce algebraic normal forms, algebraic degree, and AND/OR chains, which are a family of Boolean functions central to the implementation of comparison with constants.  We use $\#S$ to denote the cardinality of some set $S$, and we use $\ite(x, f, g) := (x \land f) \oplus (\bar x \land g)$ to denote the \emph{if-then-else} function.

\begin{definition}[Algebraic normal form]
    Let $S = \{1, \dots, n\}$ and $x_i\in\{0,1\}$.  Then
    \begin{equation}
        \label{eq:anf}
        f(x_1, \dots, x_n) = \bigoplus_{I \subseteq S} a_I\bigwedge_{i \in I}x_i
    \end{equation}
    is the \emph{algebraic normal form}~(ANF) of $f$ for some assignment to the coefficients $a_I\in\{0,1\}$.  Each AND-term in~\eqref{eq:anf}, where $a_I = 1$, is called a \emph{monomial} of $f$.
\end{definition}

\begin{example}
    \label{ex:anf}
    The ANF of $x_1 \lor x_2$ is $x_1 \oplus x_2 \oplus (x_1 \land x_2)$.  Here $a_\emptyset = 0$, but $a_{\{1\}} = a_{\{2\}} = a_{\{1,2\}} = 1$.
\end{example}

\noindent Note that every Boolean function has a unique ANF, since there are $2^n$ coefficients in~\eqref{eq:anf} and there exist $2^{2^n}$ Boolean functions over $n$ variables.

\begin{definition}[Algebraic degree]
    The \emph{algebraic degree} of a Boolean function $f$ is
    \begin{equation}
        \label{eq:degree}
        \deg(f) = \max \{\#I \mid a_I = 1\},
    \end{equation}
    where the coefficients $a_I$ are given by the ANF of $f$. In other words, the algebraic degree of $f$ is given by the number of variables in its largest monomial.
\end{definition}

\begin{definition}[Multiplicative complexity]
	Let $f$ denote a Boolean function. The multiplicative complexity of $f$, denoted by $c_\land(f)$, is the minimal number of AND gates in any logic network for $f$ over the gate set $\{\land,\oplus,\neg\}$, which consists of the 2-input AND and XOR gates and inverters.
\end{definition}

\noindent Note that OR gates can be considered AND gates in the context of the multiplicative complexity (see also Example~\ref{ex:anf}).

\begin{lemma}[Proposition 3.8, \cite{Schnorr88}]
\label{lem:degree-bound}
We have for all Boolean functions $f$, $c_\land(f) \ge \deg(f) - 1$.
\end{lemma}

\begin{definition}[AND/OR chain]
    Given Boolean variables $x_1, x_2, \dots, x_k$, $k\geq n$, an \emph{AND/OR chain} is any formula
    \begin{equation}
        \label{eq:aoc}
        f = x_1 \circ_1 (x_2 \circ_2 (\cdots (x_{n-1} \circ_{n-1} x_n) \cdots )),
    \end{equation}
    where $\circ_i \in \{\land, \lor\}$.  We refer to $\ell(f) = n - 1$ as the \emph{length} of $f$.
\end{definition}

\begin{lemma}
    \label{lem:aoc-mc}
    Let $f$ be an AND/OR chain.  Then $\deg(f) = \ell(f) + 1$ and $c_\land(f) = \ell(f)$.
\end{lemma}

\begin{proof}
    We prove that $\deg(f) = \ell(f)+1$ by induction over the length~$\ell$ of an AND/OR chain. The statement holds trivially for~$\ell = 0$.  Assuming that the statement holds for $\ell$, consider a function $x \circ f$, where $\circ \in \{\land, \lor\}$ and $f$ is an AND/OR chain and $\ell(f) = \ell$.  Further, $x$ is not in the support of $f$.  Then, we have $\deg(x \circ f) = \deg((x \oplus f)[\circ = \lor] \oplus (x \land f)) = \deg(x \land f) = 1 + \deg(f) = \ell(f) + 2$.  Since the number of operators in $f$ is $\ell(f)$, using Lemma~\ref{lem:degree-bound}, we have $c_\land(f) = \ell(f)$.
\end{proof}

\section{Construction}\label{sec:construction}
In this section, we describe an algorithm to construct a logic network to evaluate $[a \le x < b]$.  A straightforward upper bound on the number of AND gates is the sum of the costs of both individual comparisons.  However, we will present a construction that incurs at most the cost of the more costly comparison, and we show that it is possible to save an additional AND gate if both comparisons have identical costs.

\subsection{Comparison}
As a starting point, we present a construction for comparing an integer to a constant, i.e., evaluating $[a\leq x]$ for a constant integer $a$ and an $n$-bit nonnegative integer $x$. We are not aware of any previous work that describes this construction.

\begin{lemma}\label{lem:compareodd}
	Let $a = (a_1 \dots a_{n-1}1)_2$ be an odd constant integer and let $x = (x_1 \dots x_n)_2$ be an arbitrary $n$-bit nonnegative integer. Then, the AND/OR chain $x_1 \circ_1 (x_2 \circ_2 (\cdots (x_{n-1} \circ_{n-1} x_n) \cdots ))$, with
	\[
	c_i = \begin{cases} \lor & \text{if $a_i = 0$,} \\ \land & \text{if $a_i = 1$,} \end{cases}
	\]
	evaluates $[a \le x]$.
\end{lemma}

\begin{proof}
	We prove the statement using induction over $n$.  For $n = 1$, we have $a = 1$ and $x = x_1$, and $[1 \le x] = x_1$.  We assume that the statement holds for any odd constant integer of length $n$. Consider a constant $a = a_12^n + a'$ for some constant integer $a'$ of length $n$.  If $a_1 = 0$, then $[a' \le x] = x_1 \lor [a' \le (x_2 \dots x_{n+1})_2]$.  If $a_1 = 1$, then $[2^n + a' \le x] = x_1 \land [a' \le (x_2 \dots x_{n+1})_2]$.
\end{proof}

We provide pseudocode for this construction in Listing~\ref{lst:pseudo-cmp}, where we use the $\leftarrow$-operator to denote insertion of the expression on the right into the formula on the left at the unique position identified by the~$\;\cdot\;$~symbol. Invoking \texttt{comparison\_formula(a, range=\{1,n\})} produces the AND/OR chain from Lemma~\ref{lem:compareodd} for evaluating $[a\leq x]$, where $x$ is an $n$-bit integer and $a$ is an odd $n$-bit constant integer.

Note that comparing to an even number can be recast as a comparison against an odd number. This allows us to derive the multiplicative complexity of comparison using Lemma~\ref{lem:aoc-mc}, leading to the following corollary.

\begin{corollary}
	\label{lem:compare}
	Let $a = 2^jk > 0$ be a constant integer, where $k$ is odd.  Evaluating $[a \le x]$ is equivalent to evaluating $[a/2^j \le (x_1 \dots x_{n-j})_2]$, and therefore $c_\land([a \le x]) = n - j - 1$.
\end{corollary}

\definecolor{keywordcolor}{rgb}{0,0.50,0}
\begin{listing}[h]
	\begin{minted}[mathescape,bgcolor=gray!7!white,escapeinside=||,tabsize=4,]{c++}
|\textcolor{keywordcolor}{Formula}| comparison_formula(a, range):
	|$F \leftarrow (\cdot)$| // formula being constructed
	for (k = range.low; k < range.high; ++k)
		|$\circ$| = a[k] ? |$\land$| : |$\lor$| // operator is chosen according to a[k]
		|$F\leftarrow x_k \circ (\cdot) $|
	return (|$F\leftarrow x_\text{range.high} $|)
	\end{minted}
	\caption{Pseudocode for generating a formula for comparison $[a\leq x]$ (assuming $a$ is odd) using our construction. The \texttt{range} argument can be used to generate subchains, which will be useful for our interval check construction. In the code, \texttt{a[k]} denotes the $k$-th bit of the $n$-bit integer $a$, with $k$ ranging from $1$ to $n$ and \texttt{a[1]} being the most-significant bit.}
	\label{lst:pseudo-cmp}
\end{listing}

\subsection{Interval checking}
Next, we construct a logic network for interval checking. To this end, note that $[a\leq x<b]$ is equivalent to
\begin{equation}\label{eq:cmpxor}
	[a \le x] \oplus [b \le x],
\end{equation}
since $a<b$.\footnote{This expression could also be evaluated on a quantum computer using Deutsch's algorithm~\cite{Deutsch85} with a single comparator.}
Starting from this expression, our algorithm for constructing the interval check proceeds by iteratively decomposing $[a \le x] \oplus [b \le x] = (x_1 \circ f_1) \oplus (x_1 \bullet f_2)$, where both $f_1$ and $f_2$ are either AND/OR chains or constants, and $\circ$ and $\bullet$ are either $\land$ or $\lor$.  The variables involved in $f_1$ and $f_2$ are either the same and appear in the same order, or the variables in one chain are a prefix of the variables in the other chain.  We show how to evaluate this expression for different choices of $\circ$ and $\bullet$, and that this leads to a formula involving either $f_1 \oplus f_2$ (allowing us to recurse), or the if-then-else ($\ite$) operation on two AND/OR chains.

Specifically, if $f_1$ or $f_2$ is a constant, then
\begin{equation}
    \label{eq:constr-t}
    x \oplus (x \land f) = x \land \bar f
    \quad\text{or}\quad
    x \oplus (x \lor f) = \bar x \land f,
\end{equation}
and the iteration stops.  Otherwise, one of the following three cases applies:
\begin{subequations}
    \begin{align}
        \label{eq:constr-aa}
        (x \land f_1) \oplus (x \land f_2) &= x \land (f_1 \oplus f_2) \\
        \label{eq:constr-oo}
        (x \lor f_1) \oplus (x \lor f_2) &= \bar x \land (f_1 \oplus f_2) \\
        \label{eq:constr-oa}
        (x \lor f_1) \oplus (x \land f_2) &= \ite(x, \bar f_2, f_1)
    \end{align}
\end{subequations}

\begin{proof}
    Equations~\eqref{eq:constr-t} and~\eqref{eq:constr-aa} follow from straightforward Boolean identities.  For~\eqref{eq:constr-oo}, note that $(x \lor f_1) \oplus (x \lor f_2) = (\bar x \land \bar f_1) \oplus (\bar x \land \bar f_2) = \bar x \land (f_1 \oplus f_2)$, by applying De Morgan's law and using the fact that $\bar x \oplus \bar x = x \oplus x$.  By expanding the first term in~\eqref{eq:constr-oa} into an ANF, one obtains $f_1 \oplus (x \land f_1) \oplus x \oplus (x \land f_2)$.  Then the first two and the last two terms can be merged into~$(\bar x \land f_1) \oplus (x \land \bar f_2) = \ite(x, \bar f_2, f_1)$.
\end{proof}

\smallskip\noindent Next, we show a special construction for $\ite(x, \bar f_2, f_1)$ that exploits the fact that $f_1$ and $f_2$ are both AND/OR chains.  Their formulas only differ in their lengths or in what operators are used.  The idea is to propagate the inversion of $f_2$ into the formula (using De Morgan's laws) in order to make the operators match those of $f_1$.  Then, the inversions and potential chain suffixes (for chains with unequal lengths) can be implemented in the same formula conditional on $x$.

Before proving the general case, we discuss the following examples to provide some intuition.  In the first example, both chains have the same lengths and none of the operators are equal.
\begin{example}
Let $f_1 = x_2 \lor (x_3 \land x_4)$ and $f_2 = x_2 \land (x_3 \lor x_4)$.  Then
\[
    \begin{aligned}
        \bar f_2 &= \overline{x_2 \land (x_3 \lor x_4)} \\
                 &= \bar x_2 \lor \overline{(x_3 \lor x_4)} \\
                 &= \bar x_2 \lor (\bar x_3 \land \bar x_4),
    \end{aligned}
\]
and therefore
\[
    \ite(x_1, \bar f_2, f_1) = (x_1 \oplus x_2) \lor ((x_1 \oplus x_3) \land (x_1 \oplus x_4)).
\]
Note how this formula evaluates to $f_1$, if $x_1 = 0$, and to $\bar f_2$, if $x_1 = 1$.
\end{example}

\smallskip\noindent In the second example, the chain lengths are still equal, but one operator is the same: The $\lor$ operator links $x_2$ to the rest of the chain in both $f_1$ and $f_2$.
\begin{example}
Let $f_1 = x_2 \lor (x_3 \land x_4)$ and $f_2 = x_2 \lor (x_3 \lor x_4)$. Then
\[
    \begin{aligned}
        \bar f_2 &= \overline{x_2 \lor (x_3 \lor x_4)} \\
                 &= \overline{x_2 \lor \overline{\overline{(x_3 \lor x_4)}}} \\
                 &= \overline{x_2 \lor \overline{(\bar x_3 \land \bar x_4)}},
    \end{aligned}
\]
and therefore
\[
    \ite(x_1, \bar f_2, f_1) = x_1 \oplus (x_2 \lor (x_1 \oplus ((x_1 \oplus x_3) \land (x_1 \oplus x_4)))).
\]
Note how the inverter is not propagated when the operators are the same, but a double negation is introduced to further propagate the inverter to the remaining part of the chain.
\end{example}

\smallskip\noindent In the final example, we consider the case in which one chain is longer than the other.
\begin{example}
Let $f_1 = x_2 \land (x_3 \lor x_4)$ and $f_2 = x_2 \lor (x_3 \land (x_4 \lor x_5))$. Then
\[
    \begin{aligned}
        \bar f_2 &= \overline{x_2 \lor (x_3 \land (x_4 \lor x_5))} \\
                 &= \bar x_2 \land \overline{(x_3 \land (x_4 \lor x_5))} \\
                 &= \bar x_2 \land (\bar x_3 \lor \overline{(x_4 \lor x_5)}),
    \end{aligned}
\]
and therefore
\[
    \ite(x_1, \bar f_2, f_1) = (x_1 \oplus x_2) \land ((x_1 \oplus x_3) \lor (x_1 \oplus (x_4 \lor (x_1 \land x_5)))).
\]
Note how $x_1$ is not only used to invert subterms of the formula, but also to conditionally include $x_5$ to represent~$\bar f_2$.
\end{example}

\medskip\noindent We can now enumerate all cases for $f_1$ and $f_2$ in $\ite(x, \bar f_2, f_1)$.  In the following, $x_i \neq x_j$, and neither of the two variables occurs in $f_1$ or $f_2$.  The terminal cases apply when $\ell(f_1) = 0$ or $\ell(f_2) = 0$:
\begin{subequations}
\begin{align}
    \label{eq:dm-t-eq}
    \ite(x_i, \bar x_j, x_j) &= x_i \oplus x_j \\
    \label{eq:dm-t-1a}
    \ite(x_i, \bar x_j, x_j \land f) &= (x_i \oplus x_j) \land (x_i \lor f) \\
    \label{eq:dm-t-1o}
    \ite(x_i, \bar x_j, x_j \lor f) &= (x_i \oplus x_j) \lor (\bar x_i \land f) \\
    \label{eq:dm-t-2a}
    \ite(x_i, \overline{x_j \land f}, x_j) &= (x_i \oplus x_j) \lor (x_i \land \bar f) \\
    \label{eq:dm-t-2o}
    \ite(x_i, \overline{x_j \lor f}, x_j) &= (x_i \oplus x_j) \land (\bar x_i \lor \bar f)
\end{align}
\end{subequations}
In addition to these 5 terminal cases, there are 4 non-terminal cases:
\begin{subequations}
\begin{align}
    \label{eq:dm-nt-ao}
    \ite(x_i, \overline{x_j \lor f_2}, x_j \land f_1) &= (x_i \oplus x_j) \land \ite(x_i, \bar f_2, f_1) \\
    \label{eq:dm-nt-oa}
    \ite(x_i, \overline{x_j \land f_2}, x_j \lor f_1) &= (x_i \oplus x_j) \lor \ite(x_i, \bar f_2, f_1) \\
    \label{eq:dm-nt-aa}
    \ite(x_i, \overline{x_j \land f_2}, x_j \land f_1) &= x_i \oplus (x_j \land (x_i \oplus \ite(x_i, \bar f_2, f_1))) \\
    \label{eq:dm-nt-oo}
    \ite(x_i, \overline{x_j \lor f_2}, x_j \lor f_1) &= x_i \oplus (x_j \lor (x_i \oplus \ite(x_i, \bar f_2, f_1))
\end{align}
\end{subequations}
These identities are readily verified by expressing both sides of the equation as an ANF and using identities such as $x \oplus x = 0$ and $\bar x \land y = y \oplus (x \land y)$.

We provide pseudocode for our construction in Listing~\ref{lst:pseudo-ic}, where we make use of the shorthand $x^b$ for $x,b\in\{0,1\}$ to mean
\[
	x^b = \begin{cases}
				x & \text{if }b=1,\\
				\overline x & \text{if }b=0,
			   \end{cases}
\]
and for a binary operator $\circ\in\{\land,\lor\}$, we denote by $\overline{\circ}$ its dual operator, i.e., if $\circ=\land$, then $\overline{\circ}=\lor$ and vice-versa.

\begin{listing}[H]
\begin{minted}[mathescape,bgcolor=gray!7!white,escapeinside=||,tabsize=4,]{c++}
|\textcolor{keywordcolor}{Formula}| interval_check_formula(n, a, b, j_a, j_b):
	cutoff = n - min(j_a, j_b)
	|$F$| = |$(\cdot)$| // formula being constructed
	// first, remove identical prefix: Eqs. $\eqref{eq:constr-aa},\;\eqref{eq:constr-oo}$
	for (i = 1; a[i] == b[i] and i < n - max(j_a, j_b); ++i)
		|$F \leftarrow$| |$x_i^{a[i]} \land (\cdot)$|
	
	// iteration stops if $f_1$ or $f_2$ is a constant: Eq. $\eqref{eq:constr-t}$
	if i == n - j_a
		|$f$| = comparison_formula(b, range={i+1,cutoff})
		return (|$F \leftarrow$| |$x_i \land \overline{f}$|)
	if i == n - j_b
		|$f$| = comparison_formula(a, range={i+1,cutoff})
		return (|$F \leftarrow$| |$\overline{x_i} \land f$|)
	
	// handle remaining and/or-chain sections with ite()-recursion: Eqs. $\eqref{eq:dm-nt-ao}-\eqref{eq:dm-nt-oo}$
	for (j = i+1; j < n - max(j_a, j_b); ++j)
		|$\circ_j$| = a[j] ? |$\land$| : |$\lor$| // and/or is chosen according to a[j]
		if a[j] == b[j] // same operators in and/or chains
			|$F \leftarrow x_i\oplus (x_j \circ_j (x_i\oplus (\cdot)))$|
		else // different operators in and/or chains
			|$F \leftarrow (x_i\oplus x_j) \circ_j (\cdot)$|
	
	// handle terminal cases for ite(): Eqs. $\eqref{eq:dm-t-eq}-\eqref{eq:dm-t-2o}$
	if j_a != j_b // unequal lengths
		negop = j_a > j_b ? !b[j] : a[j]
		|$\circ$| = negop ? |$\land$| : |$\lor$|
		num = j_a > j_b ? |\textcolor{black}{b}| : a
		// get the postfix of the longer chain
		|$f$| = comparison_formula(num, range={j+1,cutoff})
		// and merge with the formula
		if j_a > j_b
			negop = !negop
		return (|$F \leftarrow (x_i\oplus x_j) \circ x_i^\text{negop} \; \overline{\circ}\; f^{j_a<j_b}$|)
	else // Eq. $\eqref{eq:dm-t-eq}$
		return (|$F \leftarrow x_i\oplus x_j$|)
\end{minted}
\caption{Pseudocode for generating a formula for the interval check $[a\leq x < b]$ based on the construction detailed in Section~\ref{sec:construction}. As in Listing~\ref{lst:pseudo-cmp}, \texttt{a[k]} denotes the $k$-th bit of the $n$-bit integer $a$ and \texttt{a[1]} is the most-significant bit. \texttt{j\_a} and \texttt{j\_b} correspond to $j_a$ and $j_b$, respectively.}
\label{lst:pseudo-ic}
\end{listing}

\section{Analysis}\label{sec:analysis}
In this section, we determine the multiplicative complexity of the interval check $[a \le x < b]$.  To this end, we first compute an upper bound on the multiplicative complexity by counting the number of AND and OR gates that appear in our construction from the previous section.  Then, we determine a lower bound on the multiplicative complexity by deriving the algebraic degree of our construction.  We will conclude that the lower bound matches the upper bound, allowing us to prove Theorem~\ref{thm:main}.

\subsection{Upper bound}
In the following, let $u(f)$ be the number of AND/OR gates that are applied when evaluating $f$ according to the construction discussed in the previous section.

\begin{lemma}
    \label{lem:prop-ub}
    Let $f_1$ and $f_2$ be two AND/OR chains that do not contain $x$. Then
    \[
        u(\ite(x, \bar f_2, f_1)) =
        \begin{cases}
            \ell(f_1) & \text{if $\ell(f_1) = \ell(f_2)$,} \\
            \max\{\ell(f_1), \ell(f_2)\} + 1 & \text{otherwise.}
        \end{cases}
    \]
\end{lemma}

\begin{proof}
    W.l.o.g.~we assume that $\ell(f_2) \ge \ell(f_1)$ and prove the statement by induction on $\ell(f_1)$.  In the base case, $\ell(f_1) = 0$, the statement follows from~\eqref{eq:dm-t-eq} when $\ell(f_2) = 0$.  Otherwise, either~\eqref{eq:dm-t-2a} or~\eqref{eq:dm-t-2o} applies, resulting in $2 + \ell(f) = 2 + \ell(f_2) - 1 = \ell(f_2) + 1$ AND or OR gates.  For the induction step, assume that the statement holds for $\ell(f_1) = \ell$ and let $f_1'$ be an AND/OR chain with $\ell(f_1') = \ell + 1$.  Then, one of the cases in~\eqref{eq:dm-nt-ao}--\eqref{eq:dm-nt-oo} must apply, which adds one AND or OR gate in each case and reduces the length of the formulas passed to `$\ite$' by one.
\end{proof}

\begin{lemma}
\label{lem:constr-ub}
    Let $f_1$ and $f_2$ be two distinct AND/OR chains. Then,
    \[
        u(f_1 \oplus f_2) =
        \begin{cases}
            \ell(f_1) - 1 & \text{if $\ell(f_1) = \ell(f_2)$,} \\
            \max\{\ell(f_1), \ell(f_2)\} & \text{otherwise.}
        \end{cases}
    \]
\end{lemma}

\begin{proof}
    We assume w.l.o.g.~that $\ell(f_2) \ge \ell(f_1)$.  If the first~$k$ operators (indexed from $1$~to~$k$ in~\eqref{eq:aoc}) in the AND/OR chains of $f_1$ and $f_2$ are the same, we can apply equations~\eqref{eq:constr-aa} and~\eqref{eq:constr-oo} $k$ times, resulting in $u(f_1 \oplus f_2) = k + u(f_1' \oplus f_2')$, where $f_1'$ and $f_2'$ are obtained by removing the first $k$ operators of $f_1$ and $f_2$, respectively, with $\ell(f_1') = \ell(f_1) - k$, and $\ell(f_2') = \ell(f_2) - k$.  Note that $f_1' \neq f_2'$, since $f_1 \neq f_2$.  Therefore, it is sufficient to consider the case in which the top most operators of $f_1'$ and $f_2'$ differ or the case in which $\ell(f_1') = 0$.  In the first case, $f_1'\oplus f_2'$ can be written as an if-then-else construct acting on chains that are shortened by 1 operand/operator, see \eqref{eq:constr-oa}, and the statement follows from Lemma~\ref{lem:prop-ub}.  In the second case, the statement follows from~\eqref{eq:constr-t}.
\end{proof}

\subsection{Lower bound}

We prove a lower bound on the multiplicative complexity by computing the algebraic degree of $f_1\oplus f_2$, where $f_1$ and $f_2$ are AND/OR chains. To do so, we first prove the following fact for the case where $\ell(f_1)=\ell(f_2)$.

\begin{lemma}
\label{lem:anf-terms}
Let $f_1$ and $f_2$ be two distinct AND/OR chains of identical length $n-1$. Then, the largest monomial of $f_1$ is of length $n$ and equal to that of $f_2$, and there exists a monomial of length $n-1$ that exists in only one of the two chains.
\end{lemma}
\begin{proof}
	Let $f$ denote an AND/OR chain with inputs $x_2,...,x_n$. Noting that $x_1\lor f = x_1\oplus f \oplus (x_1 \land f)$, the largest monomial of $f_1$ and $f_2$ is $(x_1 \land x_2 \land \cdots \land x_n)$.
	To prove the second statement, note that since $f_1\neq f_2$, there exists a largest position $k \leq n-1$ at which the operators in $f_1$ and $f_2$ differ. W.l.o.g., $f_1$ has an $\lor$-operator and $f_2$ has an $\land$-operator in the $k$-th position. We write $f_1=f_1'\circ (x_k\lor g(x_{k+1},\dots,x_n))$ and $f_2=f_2'\bullet (x_k\land g(x_{k+1},\dots,x_n))$, where $f_1'$ and $f_2'$ are prefixes involving variables $x_1,\dots,x_{k-1}$ and $g$ is an AND/OR chain with the largest monomial $(x_{k+1} \land \cdots \land x_n)$. In turn, this is one of the second largest monomials in the ANF of $x_k\lor g = x_k\oplus g \oplus x_k \land g$, but not in the ANF of $x_k\land g$, where all monomials feature the variable $x_k$. Consequently, note that $f_1$ features the monomial $x_1\land\cdots \land x_{k-1}\land x_{k+1}\land\cdots \land x_n$, which is not present in the ANF of $f_2$.
\end{proof}

\noindent Using Lemma~\ref{lem:anf-terms}, we can compute the algebraic degree of $f_1\oplus f_2$ as follows.

\begin{lemma}
\label{lem:constr-lb}
Let $f_1$ and $f_2$ be two distinct AND/OR chains. Then,
\[
    \deg(f_1 \oplus f_2) =
    \begin{cases}
        \ell(f_1) & \text{if $\ell(f_1) = \ell(f_2)$,} \\
        \max\{\ell(f_1), \ell(f_2)\} + 1 & \text{otherwise.}
    \end{cases}
\]
\end{lemma}

\begin{proof}
In the case $\ell(f_1) \neq \ell(f_2)$, assume w.l.o.g.~that $\ell(f_1) > \ell(f_2)$.  The degree of $f_1$ is $\ell(f_1) + 1$ (see Lemma~\ref{lem:aoc-mc}) since its single largest monomial contains all variables in the support of $f_1$. All monomials of $f_2$ are smaller, and therefore $\deg(f_1 \oplus f_2) = \deg(f_1) = \ell(f_1) + 1$.

If $\ell(f_1) = \ell(f_2)$, then Lemma~\ref{lem:anf-terms} implies that the largest monomials of $f_1$ and $f_2$ are identical and thus no longer present in the ANF of $f_1\oplus f_2$, and that there exists a monomial of length $\ell(f_1)$ that is present in only one of the two ANFs and thus also in the ANF of $f_1\oplus f_2$. Therefore, $\deg(f_1\oplus f_2) = \deg(f_1)-1 = l(f_1)$.
\end{proof}

\subsection{Proof of Theorem~\ref{thm:main}}

We are now in a position to prove the main theorem.

\begin{proof}
We have
\[
\begin{aligned}
    f=[a \le x < b] &= [a \le x] \oplus [b \le x] \\
    &= \underbrace{[(a/2^{j_a}) \le (x_1 \dots x_{n-j_a})_2]}_{f_a} \oplus \underbrace{[(b/2^{j_{b}}) \le (x_1 \dots x_{n-j_{b}})_2]}_{f_{b}} \\
    &= f_a \oplus f_{b}
\end{aligned},
\]
where $f_a$ and $f_b$ are AND/OR chains with a respective length of $n-j_a-1$ and $n-j_b-1$, see Corollary~\ref{lem:compare}. Therefore, $\max\{\ell(f_a), \ell(f_{b})\} = n - \min\{j_a, j_{b}\} - 1$.  From Lemma~\ref{lem:constr-ub}, it follows that $c_\land(f) \le n - \min\{j_a, j_{b}\} - 1 - \delta_{j_aj_{b}}$, and from Lemma~\ref{lem:constr-lb} together with Lemma~\ref{lem:degree-bound}, it follows that $c_\land(f) \ge \deg(f)-1 = n - \min\{j_a, j_{b}\} - 1 - \delta_{j_aj_{b}}$. Therefore,
\[
	c_\land(f) = n - \min\{j_a, j_{b}\} - 1 - \delta_{j_aj_{b}},
\]
where $\delta_{ij}$ denotes the Kronecker delta.
\end{proof}

\section{Conclusions}
We have derived the multiplicative complexity of interval checking given two constant bounds.  Our construction is of practical interest as it reduces the cost of interval checking by up to a factor of~2 compared to a construction composed of two comparators.  This motivates us to study the multiplicative complexity of similar composite operations, e.g., $[x = y \pm a]$ or $[a \le x + y]$, where $a$ is a constant.  This may in turn provide some insight into the multiplicative complexity of other practically-relevant operations such as multiplication, which can be considered a composition of simpler arithmetic operations.

We are further interested in studying the multiplicative complexity of formulas $f = \ite(x, f_1, f_2)$.  In general $c_\land(f) \le 1 + c_\land(f_1) + c_\land(f_2)$, but in this paper we found examples in which the multiplicative complexity of $f$ did not exceed that of the two subformulas.  Ideally, we would like to find characteristic properties for $f_1$ and $f_2$ that hold if and only if $c_\land(f) \le \max\{c_\land(f_1), c_\land(f_2)\}$.

\bibliographystyle{abbrv}
\bibliography{library}

\appendix

%

\end{document}